\documentclass[a4paper,11pt]{article}
\usepackage[utf8]{inputenc}

\usepackage{color}

\usepackage{amssymb}
\usepackage{amsmath}
\usepackage{amsfonts}
\usepackage{amsthm}
\usepackage{url}

\usepackage[dvips]{epsfig}
\usepackage[dvips]{graphics}

\usepackage{algorithm}
\usepackage{algorithmic}

\usepackage[all,knot,poly,arc]{xy}

\newcommand{\cA}{\mbox{$\cal A$}}

\newcommand{\NN}{{\Bbb N}}
\newcommand{\RR}{{\Bbb R}}

\newcommand{\ZZ}{{\Bbb Z}}

\newcommand{{\uk}}{\mbox{$\underline{k}$}}

\newcommand{\id}{\mbox{\rm id}}

\def\nod(#1,#2){\put(#1,#2){\circle*{.125}}
\put(#1,#2){\makebox(0,0.5){{\small$#2$}}}}%
\def\rod(#1,#2){\put(#1,#2){\circle*{.2}}}
\def\NOD(#1,#2)#3{\put(#1,#2){\circle*{.2}}\put(#1,#2){\makebox(0,0.8){{\small$#3$}}}}

\def\EXX{{\hfill{$\diamondsuit$}}}

\newcounter{exampleNo}

\newtheorem{theorem}{Theorem}[section]

\newtheorem{proposition}[theorem]{Proposition}
\newtheorem{corollary}[theorem]{Corollary}

\newenvironment{example}[1][Example \arabic{exampleNo}.]{\begin{trivlist}\refstepcounter{exampleNo}
\item[\hskip \labelsep {\bfseries #1}]}{\end{trivlist}}

\title{On the Decomposition of Generalized Semiautomata}
\author{Merve Nur Cakir and Karl-Heinz Zimmermann\footnote{Email: k.zimmermann@tuhh.de}\\
Department of Computer Engineering \\
Hamburg University of Technology\\
21071 Hamburg, Germany}

\begin{document}
\maketitle
\begin{abstract}
Semiautomata are abstractions of electronic devices that are deterministic finite-state machines 
having inputs but no outputs.
Generalized semiautomata are obtained from stochastic semiautomata by dropping the restrictions 
imposed by probability.
It is well-known that each stochastic semiautomaton can be decomposed into a sequential product 
of a dependent source and deterministic semiautomaton making partly use of the celebrated theorem 
of Birkhoff-von Neumann.
It will be shown that each generalized semiautomaton can be partitioned into a sequential product
of a generalized dependent source and a deterministic semiautomaton.
\end{abstract}
\medskip

\mbox{\bf AMS Subject Classification:} 68Q70, 20M35, 15A04
\medskip

\mbox{\bf Keywords:} Semiautomaton, stochastic automaton, monoid, Birkhoff-von Neumann.

\section{Introduction}

The theory of discrete stochastic systems has been initiated by the work of Shannon~\cite{shannon} and von Neumann~\cite{neumann}.
While Shannon has considered memory-less communication channels and their generalization by introducing states,
von Neumann has studied the synthesis of reliable systems from unreliable components.
The fundamental work of Rabin and Scott~\cite{rscott} about deterministic finite-state automata 
has led to two generalizations.
First, the generalization of transition functions to conditional distributions 
studied by Carlyle~\cite{carl} and Starke~\cite{starke}.
This in turn yields a generalization of discrete-time Markov chains in which the chains are governed 
by more than one transition probability matrix.
Second, the generalization of regular sets by introducing stochastic automata as described by 
Rabin~\cite{rabin}. 

By the work of Turakainen~\cite{tura69}, stochastic acceptors can be 
viewed equivalently as generalized automata in which the ''probability'' is neglected.
This leads to a more accessible approach to stochastic automata~\cite{claus}.

On the other hand, the class of nondeterministic automata~\cite{salomaa} 
can be generalized to monoidal automata,
where the input alphabet corresponds to an arbitrary monoid instead of a free monoid~\cite{kufi,mihov,zim}.
This leads to the class of monoidal automata whose languages are closed under a smaller set of operations 
when compared with regular languages.

A first step into the study of automata theory are semiautomata which are abstractions of electronic
devices that are deterministic finite-state machines having inputs but no outputs~\cite{ginzburg, mihov}.
Generalized semiautomata are obtained from stochastic semiautomata 
by dropping the restrictions imposed by probability~\cite{claus, tura69}.
It is well-known that each stochastic automaton can be decomposed into a sequential product of a dependent
source and deterministic semiautomaton~\cite{buk}.
This result makes use in part of the celebrated theorem of Birkhoff-von Neumann that each doubly stochastic
matrix can be represented as a convex combination of permutation matrices.
In this paper, it will be shown that each generalized semiautomaton can be partitioned 
into a sequential product
of a generalized dependent source and a deterministic semiautomaton.
\medskip

Notation.
Let $X$ be a set.
The set of all mappings on $X$, $T(X)=\{f\mid f:X\rightarrow X\}$, forms a monoid under function composition
$(fg)(x) = g(f(x))$, $x\in X$, and the identity function $\id_X:X\rightarrow X:x\mapsto x$ is the identity element.
The monoid $T(X)$ is called the {\em full transformation monoid\/} of $X$.

\section{Semiautomata}
Semiautomata are abstractions of electronic devices 
which are deterministic finite-state machines having input but no output~\cite{ginzburg, mihov}.

A {\em (deterministic) semiautomaton} (SA) 
is a triple $$A=(S,\Sigma,\{\delta_x\mid x\in \Sigma\})$$ 
where
\begin{itemize}
\item $S$ is the non-empty finite set of {\em states},
\item $\Sigma$ is the set of {\em input symbols},
\item $\delta_x:S\rightarrow S$ is a (partial) mapping for each $x\in \Sigma$. 
\end{itemize}

Let $\Sigma^*$ denote the free monoid over the alphabet $\Sigma$.
By the universal property of free monoids~\cite{cliff,mihov},
the mapping $\delta:\Sigma\rightarrow T(S):x\mapsto \delta_x$ 
extends uniquely to a monoid homomorphism $\delta:\Sigma^*\rightarrow T(S):u\mapsto\delta_u$ 
such that for each word $u=x_1\ldots x_k\in\Sigma^*$,
\begin{eqnarray}
\delta_u = \delta_{x_1}\cdots \delta_{x_k}
\end{eqnarray}
and particularly $\delta_\epsilon= \id_S$.
The mapping $\delta$ is called the {\em transition function\/} of~$A$.
Its image $T(A) = \{\delta_u\mid u\in \Sigma^*\}$ is a submonoid of the full transformation monoid 
$T(S)$ generated by $\{\delta_x\mid x\in \Sigma\}$.
The semiautomaton $A$ is also denoted by $A=(S,M,\delta)$ or  $A=(S^A,M^A,\delta^A)$.

A semiautomaton $A=(S,\Sigma,\delta)$ 
serves as a skeleton of a deterministic finite-state machine that is exactly in one state at a time.
If the semiautomaton $A$ is in state $s$ and reads the word $u\in \Sigma^*$, 
it transits into the state $s'=\delta_u(s)$.

\begin{example}\label{e-sa0}
Consider the semiautomaton
$A = (S,\Sigma,\delta)$ 
with state set $S=\{1,2,3\}$, input alphabet $\Sigma=\{x,y\}$, and transition function $\delta$ 
given by the automaton graph in Fig.~\ref{f-sa0}.
The associated transformation monoid is generated by the transformations
$$\delta_{x} = \left( \begin{array}{ccc} 1 & 2 & 3 \\ 1 & 1 & 1\end{array} \right)
\quad\mbox{and}\quad
\delta_{y} = \left( \begin{array}{ccc} 1 & 2 & 3 \\ 2 & 2 & 3\end{array} \right).$$
We have 
$$
\begin{array}{lll}
\delta_{xx} = \left( \begin{array}{ccc} 1 & 2 & 3\\ 1 & 1& 1\end{array} \right), &&
\delta_{xy} = \left( \begin{array}{ccc} 1 & 2 & 3\\ 2 & 2 & 2\end{array} \right),\\
\delta_{yx} = \left( \begin{array}{ccc} 1 & 2 & 3\\ 1 & 1 & 1\end{array} \right), &&
\delta_{yy} = \left( \begin{array}{ccc} 1 & 2 & 3\\ 2 & 2 & 3\end{array} \right).
\end{array}
$$
Hence, the transformation monoid $T(A)$ is given by $\{\id_S,\delta_x,\delta_y,\delta_{xy}\}$.
\EXX
\end{example}
\begin{figure}[hbt]
\begin{center}
\setlength{\unitlength}{1.0mm}
\begin{picture}(50,20)
\setlength{\unitlength}{1mm}
\put(25,10){\makebox(0,0)[c]{
\mbox{$
\xymatrix{
*++[o][F-]{1} 
 \ar@(ul,dl)[]_{x}
\ar@/^/[rr]^y
&&
*++[o][F-]{2} 
 \ar@(ur,dr)[]^{y}
\ar@/^/[ll]^x \\
& *++[o][F-]{3} 
 \ar@(ur,dr)[]^{y}
\ar@{->}[ul]^x
& \\
}
$}
}}
\end{picture}
\end{center}
\caption{Semiautomaton.}\label{f-sa0}
\end{figure}
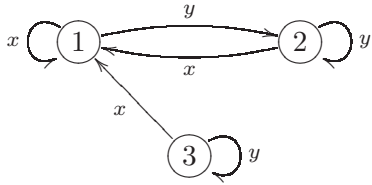

\section{Generalized Semiautomata}

Stochastic automata are a generalization of non-deterministic finite state automata~\cite{claus}.
Generalized automata can be obtained from stochastic automata by dropping the restrictions imposed by
probability~\cite{claus, tura69, zim}.

A {\em generalized semiautomaton\/} (GSA) is a triple 
$$A = (S,\Sigma,\{Q_x\mid x\in \Sigma\}),$$ 
where
\begin{itemize}
\item $S$ is the non-empty finite set of {\em states},
\item $\Sigma$ is the {\em input alphabet\/}, and
\item $Q$ is a collection of $n\times n$ nonnegative matrices $Q_x$, $x\in \Sigma$, 
where $n$ is the number of states.
\end{itemize}

In view of the universal property of free monoids~\cite{cliff,mihov},
the mapping $Q:\Sigma\rightarrow\RR^{n\times n}:x\mapsto Q_x$ extends uniquely to a monoid homomorphism
$Q:\Sigma^*\rightarrow\RR^{n\times n}$ such that for each word $u=x_1\ldots x_k\in \Sigma^*$,
\begin{eqnarray}\label{e-post0}
Q_u = Q_{x_1} \cdots Q_{x_k} 
\end{eqnarray}
and particularly $Q_{\epsilon}=I_n$ is the $n\times n$ identity matrix.
The mapping $Q$ is called the {\em transition function\/} of $A$.
Its image $T(A) = \{Q_u\mid u\in \Sigma^*\}$ is a submonoid of the full transformation monoid $T(S)$ 
generated by $\{Q_x\mid x\in \Sigma\}$.
The generalized semiautomaton $A$ is also denoted by $A=(S,\Sigma,Q)$  or $A=(S^A,\Sigma^A,Q^A)$. 

The state set $S=\{s_1,\ldots,s_n\}$ can be viewed as the standard basis for the Euclidean vector space $\RR^n$, 
where $s_i$ is the basis vector whose $i$th coordinate is~1 and all others are~0.
In this way, the $(i,j)$the entry of the matrix $Q_u=(s_{ij}^{(u)})$ is given by 
$s_{ij}^{(u)} = s_i^T Q_u s_j$.

\begin{proposition}\label{p-dg}
Each deterministic semiautomaton is a generalized automaton.
\end{proposition}
\begin{proof}
Let $A=(S,\Sigma,\delta)$ be a deterministic semiautomaton and let $S=\{s_1,\ldots,s_n\}$.
Define the generalized semiautomaton $B=(S,\Sigma,Q)$,
where for each $x\in\Sigma$, the $(i,j)$th entry of $Q_x$ is $1$ if $\delta_x(s_i)=s_j$ and otherwise~0.
Then the mapping $T(A) \rightarrow T(B):\delta_u\mapsto Q_u$ is a monoid isomorphism.
\end{proof}

%
%


A generalized semiautomaton $A = (S,\Sigma,P)$ is called {\em stochastic}
if the matrices $P_x$, $x\in\Sigma$, are stochastic, 
i.e., $P_x$ is a matrix of nonnegative real numbers such that each row sum is equal to~1.
The product of stochastic matrices is again a stochastic matrix and so the 
transition monoid $T(A)$ consists of the stochastic matrices $P_u$, $u\in\Sigma^*$.
In particular, the $(i,j)$th element $p(s_j\mid u,s_i)$ of the matrix $P_u$ is the transition probability that
the automaton enters state $s_j$ when started in state $s_i$ and reading the word $u$. 

\begin{example}
Let $m\geq 2$ be an integer.
Put $\Sigma=\{0,\ldots,m-1\}$.
The stochastic semiautomaton $\cA= (\{s_1,s_2\}, \Sigma,  P)$ given by
$$P_x = 
\frac{1}{m} \left(\begin{array}{cc} m-x & x \\ m - x -1 & x+1 \end{array}\right), \quad x\in\Sigma, 
$$
is called {\em $m$-adic semiautomaton}.
For each word $u=x_1\ldots x_k\in\Sigma^*$, 
$$P_u = 
\frac{1}{m^k} \left(\begin{array}{cc} 
m^k - w_k & w_k\\
m^k - w_k -1 & w_k+1
\end{array}\right),
$$
where $w_k =  x_km^{k-1}+\ldots + x_2m + x_1$ 
and the entry $\frac{1}{m^k}w_k$ 
corresponds in the $m$-adic representation to $0.x_k\ldots x_1$.
\EXX
\end{example}

A generalized semiautomaton $A = (S,\Sigma,D)$ is called {\em doubly stochastic}
if the matrices $D_x$, $x\in\Sigma$, are doubly stochastic, 
i.e., $D_x$ is a matrix of nonnegative real numbers such that each row and column sum is equal to~1.
The product of doubly stochastic matrices is again a doubly stochastic matrix and so the 
transition monoid $T(A)$ consists of the doubly stochastic matrices $D_u$, $u\in\Sigma^*$.


\section{Decomposition of Generalized Semiautomata}

The objective is to decompose each generalized semiautomata 
into a sequential product of a generalized dependent 
source and a deterministic semiautomaton.
The corresponding result for stochastic semiautomata has been proved by Bukharaev~\cite{buk}.

A {\em generalized dependent source} is a triple 
$$\Gamma = (\Sigma,\Xi,\{\gamma(z\mid x)\mid x\in\Sigma,z\in\Xi\}),$$ 
where
$\Sigma$ and $\Xi$ are alphabets and
$\gamma:\Sigma\times \Xi\rightarrow\RR_{\geq 0}: (x,z)\rightarrow \gamma(z\mid x)$ is a mapping 
which is extended recursively to $\Sigma^*\times\Xi^*$ as follows:
\begin{itemize}
\item $\gamma(\epsilon\mid\epsilon) = 1$, 
\item $\gamma(v\mid u) = 0$ for all $u\in\Sigma^*$ and $v\in\Xi^*$ with $|u|\ne|v|$, and
\item $\gamma(zv\mid xu) = \gamma(x\mid z)\gamma(u\mid v)$ for all 
$x\in\Sigma$, $u\in\Sigma^*$, $z\in\Xi$ and $v\in\Xi^*$.
\end{itemize}
A generalized dependent source $\Gamma$ is also denoted by $\Gamma = (\Sigma,\Xi, \gamma)$.

In particular, a {\em dependent source} is a generalized dependent source
$\Gamma = (\Sigma,\Xi,\gamma)$, 
where $\Sigma$ and $\Xi$ are alphabets and
for each $x\in\Sigma$, $\gamma(\cdot \mid x)$ defines a (conditional) probability measure on $\Xi$.
This measure can be extended for each $u\in\Sigma^*$ to a (conditional) probability measure $\gamma(\cdot \mid u)$ 
on $\Xi^*$ along the same lines as above.
Note that a dependent source can be viewed as a stochastic input-output automaton with a single 
state~\cite{buk,claus}.

The {\em sequential product} of generalized dependent source 
$\Gamma = (\Sigma,\Xi,\gamma)$ and generalized semiautomaton $B = (S,\Xi,Q^B)$ defines 
a generalized semiautomaton $A = (S,\Sigma,Q^A)$ such that for all $x\in \Sigma$,
\begin{eqnarray}
Q_x^A = \sum_{z\in\Xi} \gamma(z\mid x)\cdot Q_z^B.
\end{eqnarray}
By induction, for all $u\in \Sigma^*$,
\begin{eqnarray}
Q_u^A = \sum_{v\in\Xi^*} \gamma(v\mid u)\cdot Q_v^B.
\end{eqnarray}

A permutation matrix $P$ is a square binary matrix 
which has exactly one entry of~$1$ in each row and each column and 0's elsewhere.
By the Birkhoff-von Neumann theorem~\cite{davis}, for each $n\times n$ doubly stochastic matrix $P$ there exist 
real numbers $\alpha_1,\ldots,\alpha_N\geq 0$ with $\sum_{i=1}^N\alpha_i=1$ and permutation matrices $P_1,\ldots,P_N$ 
such that
\begin{eqnarray}
P = \alpha_1P_1+\ldots+\alpha_NP_N.
\end{eqnarray}
This representation is also known as Birkhoff-von Neumann decomposition.
Such a representation of a doubly stochastic matrix as a convex combination of permutation matrices may not be unique.
By the Marcus-Ree Theorem~\cite{ree}, $N\leq n^2-2n+2$ for dense matrices.

A square matrix $P$ is called {\em deterministic} if it has exactly one entry of~$1$ 
in each row and 0's elsewhere.
In particular, each permutation matrix is deterministic.  
For each $n\times n$ stochastic matrix $P$ there exist 
real numbers $\alpha_1,\ldots,\alpha_N\geq 0$ with $\sum_{i=1}^N\alpha_i=1$ and deterministic matrices $P_1,\ldots,P_N$
such that
\begin{eqnarray}
P = \alpha_1P_1+\ldots+\alpha_NP_N.
\end{eqnarray}
Such a representation of a stochastic matrix as a convex combination of deterministic matrices may not be unique.

A square matrix $P$ is called {\em semideterministic} 
if in each nonzero row there is exactly one entry of~$1$ and 0's elsewhere.
In particular, each deterministic matrix is semideterministic.
\begin{proposition}
For each nonnegative square matrix $A$, 
there exist real numbers $\alpha_1,\ldots,\alpha_N\geq 0$ and semideterministic matrices $P_1,\ldots,P_N$ such that
\begin{eqnarray}
A = \alpha_1P_1+\ldots+\alpha_NP_N.
\end{eqnarray}
\end{proposition}
\begin{proof}
For each nonnegative square matrix $P=(p_{ij})$ 
let $p_{i,\pi(i)}$ be a minimal nonzero entry in row $i$.
Consider the semideterministic matrix $D=(d_{ij})$ with $d_{i,\pi(i)}=1$ for each $i$ and $d_{ij}=0$ otherwise.
Moreover, put $m(P) = \min\{p_{ij}\mid p_{ij}\ne 0\}$.
Then $P-m(P)D$ is a nonnegative matrix with at least one more zero entry than $P$.
Iterating this step a finite number $N$ of times gives a sequence $(P_k)_{1\leq k\leq N}$
of nonnegative matrices and a sequence $(D_k)_{1\leq k\leq N}$
of semideterministic matrices such that 
$P_1=A$, 
$P_{k+1} = P_k - m(P_k)D_k$ for $1\leq k\leq N$, and 
$P_{N+1}=0$.
This yields the decomposition of $A$ as a linear combination of semideterministic matrices
$A = \sum_{k=1}^N m(P_k)D_k$.
\end{proof}
For doubly stochastic and stochastic matrices, the proof is similar.
\begin{example}
Consider the nonnegative matrix
$$A = \left(\begin{array}{ccc}
2 & 4 & 6\\ 2 & 2 & 8\\ 3 & 3 & 6
\end{array}\right).$$
A sequence of reductions showing the selected entries at each step is
$$
\left(\begin{array}{ccc}
\underline 2 & 4 & 6\\ \underline 2 & 2 & 8\\ \underline 3 & 3 & 6
\end{array}\right),
\left(\begin{array}{ccc}
0 & \underline 4 & 6\\ 0 & \underline 2 & 8\\ \underline 1 & 3 & 6
\end{array}\right),
\left(\begin{array}{ccc}
0 & \underline 3 & 6\\ 0 & \underline 1 & 8\\ 0 & \underline 3 & 6
\end{array}\right),
\left(\begin{array}{ccc}
0 & \underline 2 & 6\\ 0 & 0 & \underline 8\\ 0 & \underline 2 & 6
\end{array}\right),
\left(\begin{array}{ccc}
0 & 0 & \underline 6\\ 0 & 0 & \underline 6\\ 0 & 0 & \underline 6
\end{array}\right),
$$
yields the decomposition
\begin{eqnarray*}
A 
&=&
2\left(\begin{array}{ccc}
1 & 0 & 0\\ 1 & 0 & 0\\ 1 & 0 & 0
\end{array}\right)
+1\left(\begin{array}{ccc}
0 & 1 & 0\\ 0 & 1 & 0\\ 1 & 0 & 0
\end{array}\right)
+1\left(\begin{array}{ccc}
0 & 1 & 0\\ 0 & 1 & 0\\ 0 & 1 & 0
\end{array}\right)\\
&&
+\;2\left(\begin{array}{ccc}
0 & 1 & 0\\ 0 & 0 & 1\\ 0 & 1 & 0
\end{array}\right)
+6\left(\begin{array}{ccc}
0 & 0 & 1\\ 0 & 0 & 1\\ 0 & 0 & 1
\end{array}\right).
\end{eqnarray*}
\EXX
\end{example}

\begin{theorem}
Each generalized semiautomaton $A=(S,\Sigma,Q)$ 
can be represented as a sequential product of a generalized dependent source $\Gamma = (\Sigma,\Xi,\gamma)$ 
and a semideterministic semiautomaton $B = (S,\Xi,\delta)$.

In particular,
each stochastic (or strongly stochastic) semiautomaton $A=(S,\Sigma,P)$ can be represented as a 
sequential product of a dependent source $\Gamma = (\Sigma,\Xi,\gamma)$ 
and a deterministic (or permutation) semiautomaton 
$B = (S,\Xi,\delta)$.
\end{theorem}
\begin{proof}
Let $\{D_1,\ldots,D_N\}$ denote the collection of $n\times n$ semideterministic matrices.
Put $\Xi = \{1,\ldots,N\}$ and for each $x\in\Sigma$, 
write $Q_x$ as a conical combination of semideterministic matrices
$$Q_x = \sum_{z\in\Xi} \alpha(z,x) D_z.$$
This defines the generalized dependent source $\Gamma = (\Sigma,\Xi,\gamma)$, 
where for each $x\in\Sigma$ and $z\in\Xi$, $$\gamma(z\mid x) = \alpha(z,x),$$
and the deterministic automaton $B = (S,\Xi,\delta)$, where for each $z\in\Xi$,
the transition
$\delta_z:S\rightarrow S$ is given by the matrix $D_z$ as in the proof of Prop.~\ref{p-dg}. 
Then we obtain for each $x\in\Sigma$,
\begin{eqnarray*}
Q_x^A = \sum_{z\in\Xi} \gamma(z\mid x) Q_z^B.
\end{eqnarray*}
The second part is clear from the above remarks.
\end{proof}

\begin{example}
Consider the generalized semiautomaton $$A=(\{s_1,s_2\},\{x_1,x_2\},\{Q_{x_1},Q_{x_2}\}),$$ 
where
$$
Q_{x_1} = 
\left(\begin{array}{cc} 2 & 3\\ 1 & 0 \end{array}\right)
\quad\mbox{and}\quad
Q_{x_2} =
\left(\begin{array}{cc} 1 & 2\\ 0 & 3 \end{array}\right).
$$
Then 
$$
Q_{x_1} = 
\left(\begin{array}{cc} 1 & 0\\ 1 & 0 \end{array}\right)
+
\left(\begin{array}{cc} 1 & 0\\ 0 & 0 \end{array}\right)
+
3\left(\begin{array}{cc} 0 & 1\\ 0 & 0 \end{array}\right)
$$
and
$$
Q_{x_2} = 
\left(\begin{array}{cc} 1 & 0\\ 0 & 1 \end{array}\right)
+
2\left(\begin{array}{cc} 0 & 1\\ 0 & 1 \end{array}\right).
$$
Put $\Xi = \{z_1,\ldots,z_5\}$ and
$$
\begin{array}{lll}
D_{z_1} = \left(\begin{array}{cc} 1 & 0\\ 1 & 0 \end{array}\right), &
D_{z_2} = \left(\begin{array}{cc} 1 & 0\\ 0 & 0 \end{array}\right), &
D_{z_3} = \left(\begin{array}{cc} 0 & 1\\ 0 & 0 \end{array}\right), \\
D_{z_4} = \left(\begin{array}{cc} 1 & 0\\ 0 & 1 \end{array}\right), &
D_{z_5} = \left(\begin{array}{cc} 0 & 1\\ 0 & 1 \end{array}\right). &
\end{array}
$$
Then
$$ Q_{x_1} = D_{z_1} + D_{z_2} + 3 D_{z_3}
\quad\mbox{and}\quad
Q_{x_2} = D_{z_4} + 2D_{z_5}.$$
This gives the state transition table of the deterministic semiautomaton $B=(S,\Xi,\delta)$, where
$$
\begin{array}{c|ccccc}
\delta^B & z_1 & z_2 & z_3 & z_4 & z_5\\\hline
s_1    & s_1 & s_1 & s_2 & s_1 & s_2\\ 
s_2    & s_1 & -   &  -  & s_2 & s_2 
\end{array}
$$
and the transitions of the generalized dependent source $\Gamma=(\Sigma,\Xi,\gamma)$, where
$$
\begin{array}{c|ccccc}
\gamma    & z_1 & z_2 & z_3 & z_4 & z_5\\\hline
x_1 &  1  &  1  &  3  &  0  & 0 \\ 
x_1 &  0  &  0  &  0  &  1  & 2 \\ 
\end{array}
$$
\EXX
\end{example}

\begin{example}
Reconsider the $m$-adic semiautomaton $\cA=(\{s_1,s_2\},\Sigma,P)$.
For each $x\in\Sigma$,
$$P_x = 
\frac{m-x-1}{m} \left(\begin{array}{cc} 1 & 0 \\ 1 & 0 \end{array} \right)
+
\frac{1}{m} \left(\begin{array}{cc} 1 & 0 \\ 0 & 1 \end{array} \right)
+
\frac{x}{m} \left(\begin{array}{cc} 0 & 1 \\ 0 & 1 \end{array} \right).$$
Put $\Xi=\{z_1,z_2,z_3\}$ and
$$D_{z_1} = \left(\begin{array}{cc} 1 & 0 \\ 1 & 0 \end{array} \right),\quad
D_{z_2} = \left(\begin{array}{cc} 1 & 0 \\ 0 & 1 \end{array} \right),\quad
D_{z_3} = \left(\begin{array}{cc} 0 & 1 \\ 0 & 1 \end{array} \right).$$
Then for each $x\in\Sigma$,
$$P_x = 
\frac{m-x-1}{m} D_{z_1}
+
\frac{1}{m} D_{z_2}
+
\frac{x}{m} D_{z_3}.$$
This provides the state transition table of the deterministic semiautomaton $B = (S,\Xi,\delta)$, where
$$\begin{array}{c|ccc}
\delta^B & z_1 & z_2 & z_3\\\hline
s_1      & s_1 & s_1 & s_2\\
s_2      & s_1 & s_2 & s_2\\
\end{array}$$
and the transitions of the dependent source $\Gamma = (\Sigma,\Xi,\gamma)$, where
for each $x\in\Sigma$,
$$\begin{array}{c|ccc}
\gamma & z_1 & z_2 & z_3\\\hline
x     & \frac{m-x-1}{m} & \frac{1}{m} & \frac{x}{m}\\
\end{array}$$
\EXX
\end{example}

\end{document}